\documentclass{article}
\usepackage{graphicx}  

\usepackage{amsthm}
\usepackage{xcolor}  
\usepackage[margin=1in]{geometry}
\usepackage{lmodern}
\usepackage{microtype}

\usepackage{amsmath,amssymb,mathtools}

\newtheorem{fact}{Fact}
\newtheorem{theorem}{Theorem}

\newtheorem{definition}{Definition}
\newtheorem{corollary}{Corollary}

\newcommand{\TODO}{\textcolor{red}{\TODO}}

\newcommand{\C}{\mathbb{C}}
\newcommand{\Tr}{\operatorname{tr}}
\newcommand{\AME}{\mathrm{AME}}
\newcommand{\Id}{I}

\usepackage{soul}
\usepackage{hyperref}

\title{Non-existence of stabilizer absolutely maximally entangled states\\across infinitely many configurations}
\author{
    Hyunho Cha\\
    \small NextQuantum and Department of Electrical and Computer Engineering\\
    \small Seoul National University, Seoul 08826, Republic of Korea\\
    \small \texttt{ovalavo@snu.ac.kr}
}
\date{}

\begin{document}

\maketitle

\begin{abstract}
We prove a general reduction theorem for stabilizer absolutely maximally entangled states in composite local dimension. If a stabilizer $\AME(n,D)$ state exists and $D=\prod_{i=1}^m q_i$ is the prime-power factorization of $D$, then for every nonempty subset of factors there exists a stabilizer $\AME\bigl(n,\prod_{i\in M} q_i\bigr)$ state. Thus any obstruction at a prime-power factor immediately obstructs stabilizer AME states in the composite dimension.
\end{abstract}

\section{Introduction}

Absolutely maximally entangled (AME) states are pure multipartite states whose sufficiently small reduced states are maximally mixed. The case of four six-dimensional parties is particularly interesting. Recent works exhibited non-stabilizer $\AME(4,6)$ states \cite{rather2022thirty, casas2026quantum}, while the existence of a stabilizer $\AME(4,6)$ state remained unclear until \cite{gross2025thirty}, which established its non-existence in the language of 2-unitaries. We prove a stronger statement that every stabilizer AME state in composite local dimension forces stabilizer AME states on all products of its prime-power factors. The non-existence of a stabilizer $\AME(4,6)$ state then follows immediately from the known non-existence of $\AME(4,2)$ states.

\section{Preliminaries}

For a finite-dimensional Hilbert space $\mathcal{H}$, the identity operator on $\mathcal{H}$ is denoted by $\Id_{\mathcal{H}}$. If $\dim \mathcal{H}=r$, then $\Id_{\mathcal{H}}/r$ is the maximally mixed state on $\mathcal{H}$. When the underlying space is clear from context, we write $\Id_r$ for the identity on an $r$-dimensional Hilbert space.

\begin{definition}[Absolutely maximally entangled state]
Let $n,d\ge 2$. A pure state $|\Psi\rangle\in (\C^d)^{\otimes n}$ is an $\AME(n,d)$ state if for every subset $S\subseteq\{1,\dots,n\}$ with $|S|=\lfloor n/2\rfloor$, the reduced density operator on $S$ is maximally mixed:
\[
\rho_S:=\Tr_{S^c}\bigl(|\Psi\rangle\langle\Psi|\bigr)=\frac{\Id_{d^{|S|}}}{d^{|S|}}.
\]
\end{definition}

\begin{fact}
\label{fact:unitary_preserve_AME}
Local unitaries preserve the AME property.
\end{fact}

We adopt the convention for \emph{generalized Pauli operators} (also known as \emph{Weyl-Heisenberg operators}) and stabilizer states established in \cite{hostens2005stabilizer, looi2011tripartite} throughout this work.

\begin{definition}[Generalized Pauli operators, stabilizer codes, and stabilizer states]
Let \(D \ge 2\), and fix the computational basis \(\{|j\rangle\}_{j=0}^{D-1}\) of \(\mathbb{C}^D\). Define
\[
\omega := e^{\frac{2\pi i}{D}},\qquad
Z := \sum_{j=0}^{D-1} \omega^j |j\rangle\langle j|,\qquad
X := \sum_{j=0}^{D-1} |j\rangle\langle j+1|,
\]
where addition is modulo \(D\).
For \(n\) qudits, let \(X_i\) and \(Z_i\) denote the corresponding operators on the \(i\)-th tensor factor, and set
\[
\lambda := e^{\frac{2\pi i}{2D}} \qquad (\lambda^2 = \omega).
\]
A Pauli product is any operator of the form
\[
\lambda^\gamma X_1^{x_1} Z_1^{z_1} \otimes \cdots \otimes X_n^{x_n} Z_n^{z_n},
\]
where \(\gamma\) is taken modulo \(2D\) and each \(x_i,z_i\) is taken modulo \(D\). The set of all Pauli products forms the Pauli group \(P_n\).

Let \(S \subset P_n\) be an abelian subgroup consisting of linearly independent Pauli products. Its simultaneous \(+1\)-eigenspace
\[
\mathcal{C}(S) := \bigl\{\, |\psi\rangle \in (\mathbb{C}^D)^{\otimes n} : s|\psi\rangle = |\psi\rangle \text{ for all } s \in S \,\bigr\}
\]
is called the stabilizer code, and \(S\) is called its stabilizer group. The following identity always holds:
\[
|S|\,\dim \mathcal{C}(S) = D^n.
\]
Hence, if \(|S| = D^n\), then \(\dim \mathcal{C}(S)=1\). In this case \(S\) stabilizes a unique pure state, denoted by \(|S\rangle\), and \(|S\rangle\) is called a stabilizer state.
\end{definition}

\begin{theorem}[{\cite[Corollary~5]{looi2011tripartite}}]
\label{thm:looi_theorem}
Let $D=\prod_{i=1}^m q_i=\prod_{i=1}^m p_i^{e_i}$ be the prime-power factorization of $D$, where the $p_i$ are distinct primes and $q_i:=p_i^{e_i}$. For every stabilizer state $|\Phi\rangle\in (\C^D)^{\otimes n}$, there exist a unitary
\[
U:\C^D\longrightarrow \bigotimes_{i=1}^m \C^{q_i}
\]
and stabilizer states $|\phi_i\rangle\in (\C^{q_i})^{\otimes n}$ such that
\[
U^{\otimes n}|\Phi\rangle=\bigotimes_{i=1}^m |\phi_i\rangle,
\]
where the right-hand side is viewed as a state on
\[
\bigotimes_{i=1}^m (\C^{q_i})^{\otimes n}
\cong
\left(\bigotimes_{i=1}^m \C^{q_i}\right)^{\otimes n}
\cong
(\C^D)^{\otimes n}.
\]
\end{theorem}

\section{Prime-power reduction}

\begin{theorem}
\label{thm:main}
Let $n,D\ge 2$, and write the prime-power factorization of $D$ as
\[
D=\prod_{i=1}^m q_i=\prod_{i=1}^m p_i^{e_i}, \qquad q_i:=p_i^{e_i}.
\]
If a stabilizer $\AME(n,D)$ state exists, then for every nonempty subset $M\subseteq\{1,\dots,m\}$ there exists a stabilizer $\AME\bigl(n,\prod_{i\in M} q_i\bigr)$ state.
\end{theorem}

\begin{proof}
Assume that $|\Psi\rangle\in (\C^D)^{\otimes n}$ is both a stabilizer state and an $\AME(n,D)$ state. By Theorem~\ref{thm:looi_theorem}, there exist a unitary
\[
U:\C^D\longrightarrow \bigotimes_{i=1}^m \C^{q_i}
\]
and stabilizer states $|\psi_i\rangle\in (\C^{q_i})^{\otimes n}$ such that
\[
|\widetilde{\Psi}\rangle:=U^{\otimes n}|\Psi\rangle=\bigotimes_{i=1}^m |\psi_i\rangle.
\]
By Fact~\ref{fact:unitary_preserve_AME}, $|\widetilde{\Psi}\rangle$ is again an $\AME(n,D)$ state.

Fix a subset $S\subseteq\{1,\dots,n\}$ with $|S|=\lfloor n/2\rfloor$. For each $i\in\{1,\dots,m\}$, define
\[
\rho_S^{(i)}:=\Tr_{S^c}\bigl(|\psi_i\rangle\langle\psi_i|\bigr).
\]
Since $|\widetilde{\Psi}\rangle\langle\widetilde{\Psi}|=\bigotimes_{i=1}^m |\psi_i\rangle\langle\psi_i|$, its reduced state on $S$ is
\[
\widetilde{\rho}_S:=\Tr_{S^c}\bigl(|\widetilde{\Psi}\rangle\langle\widetilde{\Psi}|\bigr)=\bigotimes_{i=1}^m \rho_S^{(i)}.
\]
On the other hand, $|\widetilde{\Psi}\rangle$ is an $\AME(n,D)$ state, so
\[
\widetilde{\rho}_S=\frac{\Id_{D^{|S|}}}{D^{|S|}}.
\]
After regrouping tensor factors,
\[
\bigotimes_{k\in S}\bigotimes_{i=1}^m \C^{q_i}
\cong
\bigotimes_{i=1}^m (\C^{q_i})^{\otimes |S|},
\]
we may rewrite the maximally mixed state as
\[
\frac{\Id_{D^{|S|}}}{D^{|S|}}
=
\bigotimes_{i=1}^m \frac{\Id_{q_i^{|S|}}}{q_i^{|S|}}.
\]
Taking the partial trace over all internal factors except the $i$th one yields
\[
\rho_S^{(i)}=\frac{\Id_{q_i^{|S|}}}{q_i^{|S|}}
\qquad \text{for each } i=1,\dots,m.
\]
Because $S$ was arbitrary, every $|\psi_i\rangle$ is a stabilizer $\AME(n,q_i)$ state.

Now let $M\subseteq\{1,\dots,m\}$ be nonempty, and set
\[
d_M:=\prod_{i\in M} q_i,
\qquad
|\psi_M\rangle:=\bigotimes_{i\in M} |\psi_i\rangle.
\]
For any $S\subseteq\{1,\dots,n\}$ with $|S|=\lfloor n/2\rfloor$,
\[
\Tr_{S^c}\bigl(|\psi_M\rangle\langle\psi_M|\bigr)
=
\bigotimes_{i\in M} \rho_S^{(i)}
=
\bigotimes_{i\in M} \frac{\Id_{q_i^{|S|}}}{q_i^{|S|}}
=
\frac{\Id_{d_M^{|S|}}}{d_M^{|S|}}.
\]
Clearly, $|\psi_M\rangle$ is a stabilizer state. Hence $|\psi_M\rangle$ is a stabilizer $\AME(n,d_M)$ state.
\end{proof}

\begin{figure}
    \centering
    \includegraphics[width=0.9\linewidth]{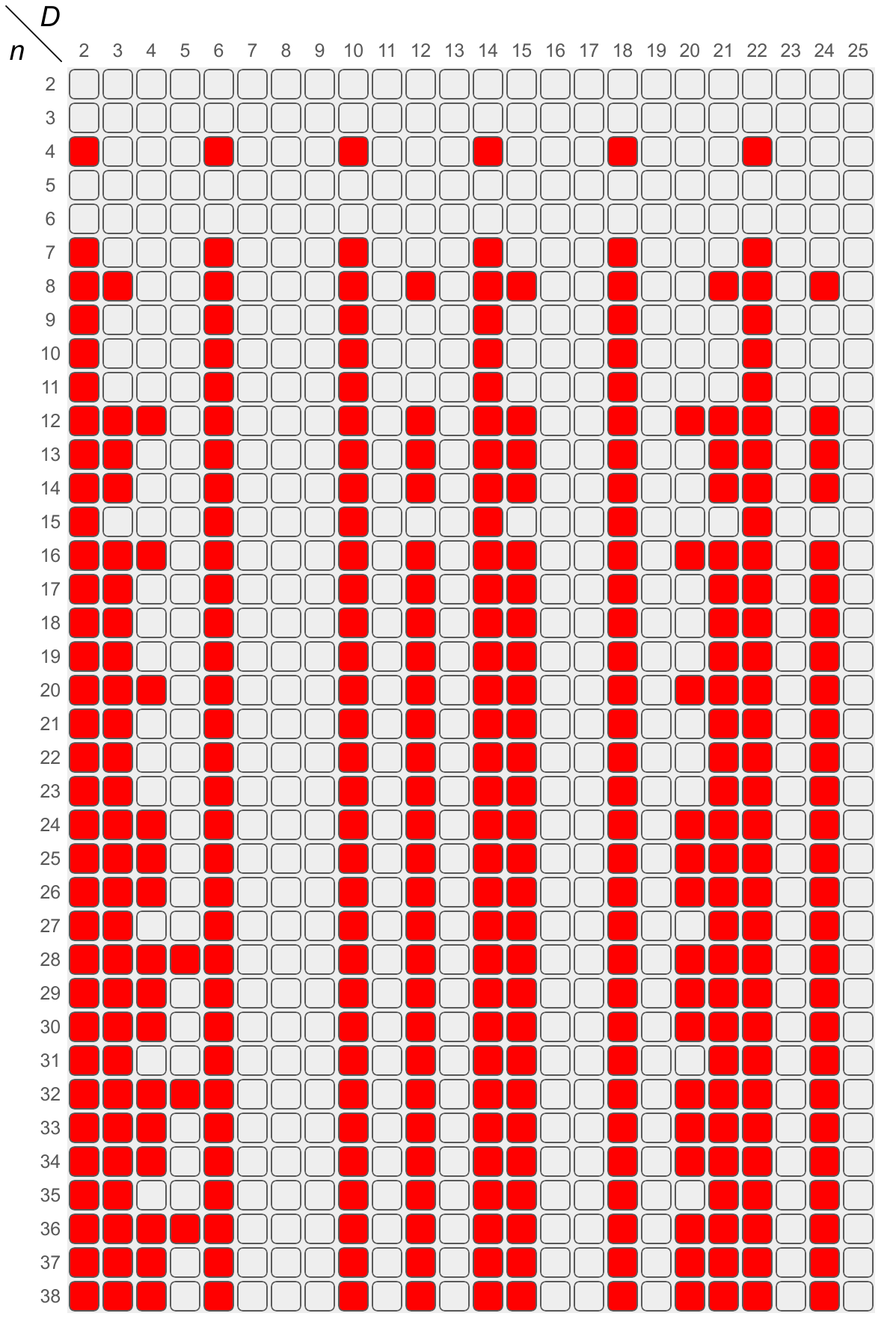}
    \caption{A red $(n,D)$ square indicates that a stabilizer $\AME(n,D)$ state is guaranteed not to exist.}
    \label{fig:NoStabAME}
\end{figure}

This theorem can be applied to any established non-existence result for lower dimensions, such as those documented in \cite{huber2020table}.
Figure~\ref{fig:NoStabAME} illustrates configurations for which stabilizer AME states are excluded.

\section{Four-party consequences}

Theorem~\ref{thm:main} can be specialized to four parties, but we remark that the following results are already latent within the framework of \cite{gross2025thirty}.

\begin{theorem}[{\cite[Theorem~1]{higuchi2000entangled}}]
No $\AME(4,2)$ state exists.
\end{theorem}

\begin{corollary}
There is no stabilizer $\AME(4,D)$ state for any $D\equiv 2\pmod 4$. In particular, there is no stabilizer $\AME(4,6)$ state.
\end{corollary}

\section{Discussion}

The reduction theorem presented here provides a clear demarcation between the existence of general AME states and their stabilizer counterparts. From a practical perspective, this result serves as a ``no-go'' theorem that streamlines future research. It should be noted, however, that our result is established within the specific stabilizer convention adopted in this work. Although theoretically negative, they ensure that researchers do not expend effort on impossible constructions. Understanding where the stabilizer formalism fails is as crucial as understanding where it succeeds, as it defines the boundary where more complex quantum resources must be invoked.






\bibliographystyle{unsrt}
\bibliography{main}

\end{document}